\documentclass[submission,copyright,creativecommons]{eptcs}
\providecommand{\event}{WPTE17} 
\usepackage{breakurl}             
\usepackage{amsmath}
\usepackage{amsthm}
\usepackage{amssymb}
\usepackage{esvect}
\usepackage{paralist}
\usepackage[all]{xy}
\usepackage{tikz}
\usetikzlibrary{positioning,arrows}
\usepackage{pgfplots}

\title{Space Improvements and  Equivalences in a Functional Core Language}
\author{Manfred Schmidt-Schau{\ss}\thanks{supported by the Deutsche Forschungsgemeinschaft (DFG) under grant SCHM 986/11-1.}
\institute{Goethe-University\\Frankfurt am Main}
\email{schauss@ki.cs.uni-frankfurt.de}
\and
 Nils Dallmeyer\footnotemark[1]
\institute{Goethe-University\\Frankfurt am Main}
\email{dallmeyer@ki.cs.uni-frankfurt.de} }
\def\titlerunning{Space Improvements and  Equivalences in a Functional Core Language}
\def\authorrunning{M. Schmidt-Schau{\ss} \& N. Dallmeyer}

\newcommand{\ignore}[1]{}

\newcommand{\ignoredavid}[1]{}

\newcommand{\paper}[1]{}

\newcommand{\mln}{\ensuremath{\text{\normalfont\ttfamily mln}}}

\newcommand{\size}{\mathtt{size}}

\newcommand{\rln}{{\tt rln}}

\newcommand{\rlnall}{{\tt rlnall}}

\newcommand{\ari}{\mathit{ar}}

\newcommand{\tletrx}[2]{({\tletr~#1~\tin~#2})}

\newcommand{\tletr}{{\tt letrec}}
\newcommand{\tletrec}{\tletr}
\newcommand{\tlet}{{\tt let}}
\newcommand{\tin}{{\tt in}}
\newcommand{\tif}{{\tt if}}
\newcommand{\tthen}{{\tt then}}
\newcommand{\telse}{{\tt else}}
\newcommand{\ttrue}{{\tt True}}

\newcommand{\tnil}{{\tt Nil}}

\newcommand{\tcase}{{\tt case}}
\newcommand{\tof}{{\tt of}}

\newcommand{\tseq}{{\tt seq}}

\newcommand{\iEnv}{{\mathit{E}}}

\newcommand{\casepf}{\,\texttt{->}\,}

\newcommand{\LR}{\mathit{LR}}
\newcommand{\LRP}{\mathit{LRP}}

\newcommand{\FV}{{\mathit{FV}}}

\newcommand{\mustdiv}{{\uparrow}}
\newcommand{\maycon}{{\downarrow}}

\newcommand{\bchainGen}[6]{\{{#1_{#2}=#3_{#4}}\}_{i=#5}^{#6}}

\newcommand{\bchainN}[3]{\bchainGen{#1}{i}{#2}{i}{1}{#3}}

\newcommand{\eg}{{{e.g.}}}

\newcommand{\tletrxnk}[2]{\tletr~#1 ~{\tt in}~#2}

\newcommand{\bbbR}{\mathbb{R}}

\newenvironment{proof*}{{\it Proof.}}{}
\newtheorem{theorem}{Theorem}[section]
\newtheorem{lemma}[theorem]{Lemma} 
\newtheorem{example}[theorem]{Example} 
\newtheorem{definition}[theorem]{Definition } 
\newtheorem{remark}[theorem]{Remark} 
\newtheorem{proposition}[theorem]{Proposition}
\newtheorem{corollary}[theorem]{Corollary} 

\usepackage{esvect}
\usepackage[all]{xy}

\begin{document}
\maketitle
\begin{abstract}
We explore space improvements in $\LRP$, a polymorphically typed call-by-need functional core language. 
A relaxed space measure is chosen for the maximal size usage during an evaluation. It
 abstracts from the details of the implementation via abstract machines, but 
it takes garbage collection into account and thus can be seen as a  realistic approximation of space usage.
 The results are: a context lemma for space improving translations and for space equivalences; 
 all but one reduction rule of the calculus are shown to be space improvements,  
 and for the exceptional one we show bounds on the space increase.
 Several further program transformations are shown to be space improvements or space equivalences,
 in particular the translation into machine expressions is a  space equivalence.
 We also classify certain space-worsening transformations as space-leaks or as space-safe.
 These results are a step forward in making predictions about the change in runtime space behavior of optimizing transformations 
 in call-by-need functional languages.
\end{abstract}


\newcommand{\ignoreForAbstract}[1]{}

\newcommand{\itspace}{\mathit{space}}
\newcommand{\spmax}{\mathit{spmax}}
\newcommand{\talt}{{\tt alt}}
\newcommand{\LRPgc}{\mathit{LRPgc}}
\newcommand{\LRgc}{\mathit{LRgc}}

\newcommand{\foldl}{{\tt foldl}}
\newcommand{\foldls}{{\tt foldl'}}
\newcommand{\foldr}{{\tt foldr}}

\newcommand{\itsumspace}{\mathit{sumspace}}
\newcommand{\spsum}{\mathit{spsum}}
\newcommand{\LCSC}{\mathit{LCSC}}
\newcommand{\Red}{\mathit{Red}}
\newcommand{\WHNF}{\mathit{WHNF}}
\newcommand{\synsize}{\mathtt{synsize}}
\newcommand{\T}{{T}}
\newcommand{\SC}{{S}}
\newcommand{\TT}{\T\T}
\newcommand{\seq}{{\tt seq}\ }
\newcommand{\nogc}{\mathit{LRPgc}}

\newcommand{\ltop}{\texttt{top}}
\newcommand{\lsub}{\texttt{sub}}
\newcommand{\lvis}{\texttt{vis}}
\newcommand{\lnontarg}{\text{nontarg}}
\newcommand{\lvee}{\hspace*{1.25pt}\vee\hspace*{1.25pt}}

\newcommand{\maydiv}{{\uparrow}}
\newcommand{\maycongc}{{\downarrow}_{\nogc}}

\newcommand{\AALL}{\mathfrak{A}}
\newcommand{\rlnnull}{\rln_{\AALL}}
\newcommand{\mlnnull}{\mln_{\AALL}}
\newcommand{\gceq}{{gc{=}}}  
\newcommand{\alts}{\mathit{alts}}  

\newcommand{\tfoldl}{{\tt foldl}}
\newcommand{\hlinep}{\\[-3mm]\hline\\[-3mm]}

%
%
%
%
%

\section{Introduction}

The focus of this paper is on providing methods for analyzing optimizations for call-by-need functional languages. 
Haskell \cite{haskell:2010,haskell-web:2016} is a functional programming language that uses lazy evaluation, and employs a polymorphic type system.
Programming in Haskell is declarative, which avoids overspecifying imperative details of the actions at runtime. 
Together with the type system this leads to a compact style of high level programming and avoids several types of errors.

The declarative features must be complemented with a more sophisticated compiler including  optimization methods and procedures. 
Declarativeness in connection with lazy evaluation (which is call-by-need \cite{ariola:95,moran-sands-carlsson:99} as a sharing variant of call-by-name) gives a lot of freedom to the exact execution 
and is accompanied by a semantically 
founded notion of correctness of the compilation. Compilation is usually a process that translates the surface program into a core language, 
where the optimization process can be understood as a sequence of transformations producing a final program.

Evaluation of a program or of an expression in a lazily evaluating functional  language is connected with variations in the evaluation sequences of  
the expressions in function bodies, depending on the arguments.  
Optimization exploits this and usually leads to faster evaluation. 
The easy-to-grasp notion of time improvements is contrasted by an opaque behavior of the evaluation w.r.t.\ space usage, 
which in the worst case may lead to space leaks  (high space usage during evaluation, which perhaps could be avoided by correctly transforming the program before evaluation).
Programmers may experience space leaks as unpredictability of space usage, 
generating rumors like ``Haskell's space 
usage prediction is a black art'' and in fact a loss of trust into the optimization. 
\cite{gustavsson-sands:99,gustavsson-sands:01,bakewell-runciman-ppdp:2000} observed that semantically correct modifications of the sequence of evaluation 
(for example due to strictness information) may have a dramatic effect on space usage, where an example is $(\texttt{head}~ xs)~ \texttt{eqBool}~ (\texttt{last}~ xs)$  vs. 
$(\texttt{last}~ xs)~\texttt{eqBool}~ (\texttt{head}~ xs)$ where $xs$ is bound to an expression that generates a long list of Booleans (using the  Haskell-conventions).

Early work on space improvements by Gustavsson and Sands \cite{gustavsson-sands:99,gustavsson-sands:01} provides deep insights and founded methods     
to analyze the dynamic space usage of programs in call-by-need functional languages. Our work is a reconsideration of the same issues, but there are some 
 differences:  their calculus has a restricted syntax (for example the arguments of function calls must be variables), whereas our calculus is unrestricted; 
 they investigate an untyped calculus, whereas we investigate a typed calculus. 
 Measuring space is also slightly different: whereas \cite{gustavsson-sands:99,gustavsson-sands:01} counts only the heap bindings, we count the whole expression, 
 but instead omit parts of the structure 
 (for example variables are not counted). The difference in  space measuring appears to be arbitrary, however,  our measure turns out to ignore the right amount of 
 noise and subtleties of space behavior, but nevertheless sufficiently models the reality, and leads to general and good estimations.
 


%
  %
The focus of this paper is to contribute to a better understanding of the space usage of lazy functional languages and to enable tools
for a  better predictability of space requirements. The approach is to analyze a polymorphically typed and call-by-need
 functional core language 
$\LRP$
that is a lambda calculus extended with the constructs  letrec, case, constructors, seq, type-abstraction, and with call-by-need evaluation.
 Call-by-need evaluation has a sharing regime and due to the recursive bindings by a letrec, in fact a sharing graph is used.
   Focusing on space usage enforces to  include garbage collection 
 into the model, i.e. into the core language. 
 This model is our calculus $\LRPgc$.

 The {\bf contributions and results}\ of this paper are: a definition (Def.\ \ref{def:space-improvement}) of the space measure $\spmax$ as an abstract version of the maximally 
 used space by an abstract machine during an evaluation, and a corresponding definition of transformations to be max-space-improvements 
 or -equivalences, where the criterion is that this holds in every context. 
   A context lemma  (Prop.\ \ref{lemma:context-sp-impr}) is proved that eases proofs of transformations being space improvements or equivalences. 
 The main result is a classification in Sect. \ref{sec:space-safe-transformations}   
  of the rules of the calculus (but one)
 used as transformations, and 
  of further transformations as max-space improvements
 and/or max-space equivalences, or as increasing max-space, or even as space-leaks.  
 These results   
 imply that the transformation into machine expressions keeps the max-space usage 
  which also holds for the evaluations on the abstract machine.  
  We  also classify    
       some space-worsening transformations as well-behaved (space-safe up to) or as space-leaks.
  We also argue that the typed calculus has more improvements than its untyped version.
   This is a contribution to predicting the space behavior of optimizing
 transformations, which in the future may lead also to a better control of powerful, but space-hazardous, time-optimizing transformations.
 
  
We discuss some  {\bf previous work} on time and space behavior for call-by-need functional languages.  
 Examples of research on 
the correctness of program transformations are in ~\cite{ghcinliner02,johann-voigtlaender:06,schmidt-schauss-schuetz-sabel:08},  
 examples of the use of transformations in optimization in functional languages are in \cite{peyton-santos:98,Sculthorpe-Farmer-Gill:13:HERMIT}.
 A theory of (time) optimizations of call-by-need
functional languages was started in \cite{moran-sands:99} for a call-by-need  higher order language, 
also based on a variant of Sestoft's abstract machine \cite{sestoft:97}.
An example transformation with a high  potential to
improve efficiency is common subexpression elimination, which  is considered in \cite{moran-sands:99}, but not proved to be a time
improvement (but it is conjectured), and which  is  proved correct in \cite{schmidtschauss-sabel-PPDP:2015}, and proved in this paper as space leak.
 Hackett and Hutton \cite{hackett-hutton:14} applied the improvement theory of \cite{moran-sands:99}
to argue that optimizations are indeed improvements, with a particular focus on  (typed) worker/wrapper transformations 
(see e.g. \cite{birdbook:14} for more examples).
The work of \cite{hackett-hutton:14} uses the same call-by-need abstract machine as \cite{moran-sands:99} with a slightly modified measure
for the improvement relation.
Further work that analyses space-usage of a lazy functional language is \cite{bakewell-runciman-ppdp:2000}, for a language without letrec and using  a term graph model,
 and comparing different evaluators.

The {\bf structure of the paper} is to first define the calculi $\LRP$ in Sect.\  \ref{subsec:lrp-def} and a variant  $\LRPgc$ with garbage collection
 in Sect.\  \ref{sec:calculi-gc}.
Sect.\   \ref{sec:space-improvements} defines space improvements and contains the  context lemmata.
Sect.\   \ref{sec:space-safe-transformations} discusses space-safeness and space-leaks, and
   contains a detailed treatment of space improving transformations, and discusses specific examples.
 Sect.\   \ref{sec:inlining} contains experiments measuring space- and time-usage for an inlining transformation, 
 which cannot be derived from the current theory. 
Missing  explanations, arguments and proofs can be found in the technical report \cite{schmidt-schauss-dallmeyer:2017:frank-57}.   

\begin{figure*}[t]   
 \fbox{ 
\begin{minipage}{.98\textwidth}
\noindent{\bf Syntax of expressions and types:}
Let type variables $a,a_i\in\mathit{TVar}$ and term variables $x_,x_i\in\mathit{Var}$. 
Every type constructor $K$ has an arity $ar(K)\geq 0$ and a finite set $D_K$ of data constructors $c_{K,i}\in D_K$ 
with an arity $ar(c_{K,i})\geq 0$. \\[2mm]
\textbf{Types} $Typ$ and polymorphic types $\mathit{PTyp}$ are defined as follows:\\
$\begin{array}{lll}
\tau\in\mathit{Typ}  & ::= & a \mid (\tau_1 \rightarrow \tau_2) \mid (K\ \tau_1\ \dots\ \tau_{ar(K)})\\
\rho\in\mathit{PTyp} & ::= & \tau \mid \forall a.\rho
\end{array}$\\[2mm]
\textbf{Expressions} $\mathit{Expr}$ are generated by this grammar with $n\geq 1$ and $k\geq 0$:\\
$\begin{array}{lll}
r,s,t\in\mathit{Expr} & ::= & u \mid x{::}\rho \mid (s\ \tau) \mid (s\ t) \mid (\seq s\ t) \mid (c_{K,i}{::}(\tau)\ s_1\ \dots\ s_{ar(c_{K,i})})  \\   & &
            \mid (\tletr~ x_1{::}\rho_1=s_1,\dots,x_n{::}\rho_n=s_n~\tin~ t)\\
                      & & \mid (\tcase_K\ s~\tof~\{(Pat_{K,1} \casepf t_1)\ \dots\ (Pat_{K,|D_K|} \casepf t_{|D_K|})\})\\
Pat_{K,i}           & ::= & (c_{K,i}::(\tau)\ (x_1::\tau_1)\ \dots\ (x_{ar(c_{K,i})}::\tau_{ar(c_{K,i})}))\\
u\in\mathit{PExpr}  & ::= & (\Lambda a_1.\Lambda a_2.\dots\Lambda a_k.\lambda x::\tau.s)
\end{array}$
\caption{Syntax of expressions and types of LRP\label{fig:lrp}}
\end{minipage}
 }
\end{figure*} 

\begin{figure*}[ht]  {\small 
 \fbox{ \hspace*{-1mm}
\begin{minipage}{.98\textwidth}
$\begin{array}{@{}ll@{}}
\text{(lbeta)}   & ((\lambda x.s)^\lsub\ r) \rightarrow (\tletr~x=r~\tin~ s)\\
\text{(Tbeta)}   & ((\Lambda a.u)^\lsub\ \tau) \rightarrow u[\tau/a]\\
\text{(cp-in)}   & (\tletr~x_1=v^\lsub,\{x_i=x_{i-1}\}_{i=2}^m,\iEnv~\tin~ C[x_m^\lvis]) 
                  \rightarrow (\tletr~x_1=v,\{x_i=x_{i-1}\}_{i=2}^m,\iEnv~\tin~ C[v])\\
                 & \qquad \text{where }v\text{ is a polymorphic abstraction}\\
\text{(cp-e)}    & (\tletr~x_1=v^\lsub,\{x_i=x_{i-1}\}_{i=2}^m,\iEnv,y=C[x_m^\lvis]~\tin~ r) \\
                  & \qquad  \rightarrow (\tletr~x_1=v,\{x_i=x_{i-1}\}_{i=2}^m,\iEnv,y=C[v]~\tin~ r)\\
                 & \qquad \text{where }v\text{ is a polymorphic abstraction}\\
\text{(llet-in)} & (\tletr~\iEnv_1~\tin~ (\tletr~\iEnv_2~\tin~ r)^\lsub) \rightarrow (\tletr~\iEnv_1,\iEnv_2~\tin~ r)\\
\text{(llet-e)}  & (\tletr~\iEnv_1,x=(\tletr~\iEnv_2~\tin~ t)^\lsub~\tin~ r) \rightarrow (\tletr~\iEnv_1,\iEnv_2,x=t~\tin~ r)\\
\text{(lapp)}    & ((\tletr~\iEnv~\tin~ t)^\lsub\ s) \rightarrow (\tletr~\iEnv~\tin~ (t\ s))\\
\text{(lcase)}   &   (\texttt{case}_K\ (\tletr~\iEnv~\tin~ t)^\lsub~\tof~ alts) \rightarrow (\tletr~\iEnv~\tin~ (\texttt{case}_K\ t~\tof~ alts))\\
\text{(lseq)}    & (\seq (\tletr~\iEnv~\tin~ s)^\lsub\ t) \rightarrow (\tletr~\iEnv~\tin~ (\seq s\ t))\\
\text{(seq-c)}   & (\seq v^\lsub\ t) \rightarrow t\qquad \text{if }v\text{ is a value}\\
\text{(seq-in)}  & (\tletr~x_1=(c\ \vv{s})^\lsub,\{x_i=x_{i-1}\}_{i=2}^m,\iEnv~\tin~ C[(\seq x_{m}^\lvis\ t)]) \\
                   & \qquad  \rightarrow (\tletr~x_1=(c\ \vv{s}),\{x_i=x_{i-1}\}_{i=2}^m,\iEnv~\tin~ C[t])\\
\text{(seq-e)}   & (\tletr~x_1=(c\ \vv{s})^\lsub,\{x_i=x_{i-1}\}_{i=2}^m,\iEnv,y=C[(\seq x_{m}^\lvis\ t)]~\tin~ r) \\
                  & \qquad   \rightarrow (\tletr~x_1=(c\ \vv{s}),\{x_i=x_{i-1}\}_{i=2}^m,\iEnv,y=C[t]~\tin~ r)\\
\text{(case-c)}  & (\tcase_K\ c^\lsub~\tof~ \{ \dots (c \casepf t) \dots\} ) \rightarrow t \quad \text{if }ar(c)=0\text{, otherwise:}\\
                  &  (\tcase_K\ (c\ \vv{s})^\lsub~\tof~ \{ \dots ((c\ \vv{y}) \casepf t) \dots \}) \rightarrow (\tletr~\{y_i=s_i\}_{i=1}^{ar(c)}~\tin~ t)\\
\text{(case-in)} &  (\tletr~x_1=c^\lsub, \{x_i=x_{i-1}\}_{i=2}^m,\iEnv~\tin~ C[(\tcase_K\ x_m^\lvis~\tof~ \{(c\casepf r)\dots\})])\\
                  &  \qquad \rightarrow (\tletr~x_1=c, \{x_i=x_{i-1}\}_{i=2}^m,\iEnv~\tin~ C[r])\qquad \text{ if }ar(c)=0;~\text{otherwise:}\\
                  &  (\tletr~x_1=(c\ \vv{t})^\lsub,\{x_i=x_{i-1}\}_{i=2}^m,\iEnv~\tin~ C[(\tcase_K\ x_m^\lvis~\tof~ \ \{((c\ \vv{z})\casepf r) \dots \})])\\
                 & \qquad \rightarrow (\tletr~x_1=(c\ \vv{y}),\{y_i=t_i\}_{i=1}^{ar(c)}, \{x_i=x_{i-1}\}_{i=2}^m,\iEnv~  \tin~ C[\tletr~\{z_i=y_i\}_{i=1}^{ar(c)}~\tin~ r])\\
\text{(case-e)}  & (\tletr~x_1=c^\lsub, \{x_i=x_{i-1}\}_{i=2}^m, u=C[(\tcase_K\ x_m^\lvis~\tof~ \{(c\casepf r_1)\dots\})], 
                   \enskip \iEnv~\tin~ r_2)\\
                  & \qquad \rightarrow (\tletr~x_1=c, \{x_i=x_{i-1}\}_{i=2}^m, u=C[r_1],\iEnv~\tin~r_2)~\text{ if }ar(c)=0 ;~\text{otherwise:}\\
                  & (\tletr~x_1=(c\ \vv{t})^\lsub,\{x_i=x_{i-1}\}_{i=2}^m, 
                     u=C[(\tcase_K\ x_m^\lvis~\tof~ \{ \dots((c\ \vv{z}) \casepf r)\dots \})],\iEnv~\tin~ s)\\
                 & \qquad \rightarrow (\tletr~x_1=(c\ \vv{y}), \{y_i=t_i\}_{i=1}^{ar(c)}, \{x_i=x_{i-1}\}_{i=2}^m, 
                  u=C[\tletr~\{z_i=y_i\}_{i=1}^{ar(c)}~\tin~ r],\iEnv{~\tin~} s)\\
\end{array}$\\
     The variables $y_i$ are fresh ones in (case-in) and (case-e).
\caption{Basic $\LRP$-reduction rules \cite{schmidt-schauss-sabel-WPTE:14} \label{fig:basred}}
\end{minipage}}}
\end{figure*}
 
\section{Polymorphic and Untyped Lazy Lambda Calculi} \label{sec:lrp}
We introduce the polymorphically typed calculus $\LRP$, and the variant $\LRPgc$ with garbage collection,
since numerous complex transformations have their nice space improving property under all circumstances
(in all contexts) only in a typed language. Technically, this shows up in the proofs when we have to argue over all contexts, 
which are strictly less than without types. For example,  case analyses have to inspect less cases, 
in particular for list-processing functions (i.e. a smaller number and simpler forking diagrams). 

 
\subsection{$\LRP$: The Polymorphic Variant}\label{subsec:lrp-def}

Let us recall  the polymorphically typed and extended lazy lambda calculus ($\LRP$) 
\cite{sabel-schmidt-schauss:2015:frank-55,schmidtschauss-sabel-PPDP:2015,schmidt-schauss-sabel-WPTE:14,schmidt-schauss-sabel:frank-56:15}.
We  motivate and 
introduce 
several necessary extensions of $\LRP{}$ which support realistic space analyses.

$\LRP{}$ \cite{schmidt-schauss-sabel-WPTE:14} is $\LR{}$ (an extended call-by-need lambda calculus with \tletrec, 
\eg\  see \cite{schmidt-schauss-schuetz-sabel:08}) extended with types. 
I.e. $\LRP{}$ is an extension of the lambda calculus by polymorphic types, 
recursive \texttt{letrec}-expressions, \texttt{case}-expressions, \texttt{seq}-expressions, data constructors, polymorphic abstractions $\Lambda a.s$ 
to express polymorphic functions and  type applications $(s\ \tau)$ for type instantiations. The syntax of expressions and types of $\LRP{}$ is 
defined in Fig.~\ref{fig:lrp}.

An expression is well-typed if it can be typed using typing rules that are defined in \cite{schmidt-schauss-sabel-WPTE:14}. 
%
$\LRP{}$ is a core language of Haskell and is simplified compared to Haskell, because it does not have type classes and is only polymorphic in the bindings of {\tletr} variables. 
But $\LRP{}$ is sufficiently expressive for polymorphically typed lists and functions working on such data structures. 

From now on we use $\iEnv$ as abbreviation for a multiset of bindings of the form $x = e$, also called  \texttt{letrec}-environment.
We also use  $\{x_{g(i)}=s_{f(i)}\}_{i=j}^m$ for $x_{g(j)}=s_{f(j)},\dots,x_{g(m)}=s_{f(m)}$ and $alts$ for
 \texttt{case}-alternatives. Bindings in \tletr-environments can be commuted. We use $FV(s)$ and $BV(s)$ to denote free and bound variables of an expression 
 $s$, $LV(\iEnv)$ to denote the binding variables of a
  \texttt{letrec}-environment, and  we abbreviate $(c_{K,i}\ s_1\ \dots\ s_{ar(c_{K,i})})$ with $c\ \vv{s}$ and $\lambda x_1.\dots\lambda x_n.s$ with 
  $\lambda x_1,\dots,x_n.s$. The data constructors \texttt{Nil} and \texttt{Cons} are used to represent lists, but we may also use the Haskell-notation \texttt{[]} and (:) instead. 
A {\em context} $C$ is an expression with exactly one (typed) hole $[\cdot_\tau]$ at expression position. 
A {\em surface context}, denoted $\SC$, 
is a context where the hole is not within an abstraction,
and a {top context}, denoted $\T$, is a context where the hole is not in an abstraction nor in a case-alternative.
A {\em reduction context} is a context where reduction may take place, and it is defined using a labeling algorithm that indicates 
the call-by-need reduction positions 
\cite{schmidt-schauss-sabel-WPTE:14}.
Reduction contexts are for example  $[\cdot]$, $([\cdot]~e)$, $(\tcase~[\cdot]~ \ldots)$ and $\tletr~x = [\cdot], y = x,\ldots~\tin~(x~\ttrue)$.
Note that reduction contexts are surface as well as top-contexts.
A {\em value} is an abstraction $\lambda x.s$, 
a polymorphic abstraction $u$ or a constructor application $c\ \vv{s}$. 

We explain the rules in Fig.\ \ref{fig:basred}.  The classical $\beta$-reduction is replaced by the sharing (lbeta). 
(Tbeta) is used for type instantiations concerning polymorphic type bindings. The rules  (cp-in) and (cp-e) copy abstractions which are needed 
when the reduction rules have to reduce an application $(f\ a)$ where $f$ is an abstraction defined in a \texttt{letrec}-environment. 
The rules (llet-in) and (llet-e) are used to merge nested \texttt{letrec}-expressions;
 (lapp), (lcase) and (lseq) move a \texttt{letrec}-expression out of an application, a \texttt{seq}-expression or a \texttt{case}-expression;
  (seq-c), (seq-in) and (seq-e) evaluate \texttt{seq}-expressions, where the first argument has to be a value or a value which is reachable through 
  a \texttt{letrec}-environment.
  (case-c), (case-in) and (case-e) evaluate 
  \texttt{case}-expressions by using \texttt{letrec}-expressions to realize the insertion of the variables for the appropriate
  \texttt{case}-alternative.

The following abbreviations are used:  (cp) is the union of (cp-in) and (cp-e);
 (llet) is the union of (llet-in) and (llet-e);
  (lll) is the union of (lapp), (lcase), (lseq) and (llet);
   (seq) is the union of (seq-c), (seq-in), (seq-e);
  (case) is the union of (case-c), (case-in), (case-e).

\begin{definition}[Normal order reduction] \label{def:redstep}
A {\em normal order reduction step} $s\xrightarrow{\LRP} t$ is performed (uniquely) if the (top-down) labeling algorithm in
 \cite{schmidt-schauss-sabel-WPTE:14} terminates 
 on 
 $s$ inserting the (superscript) labels $\lsub$ (subexpression) and $\lvis$ (visited) 
  and the applicable rule (i.e. matching also the labels) of
 Fig.\ \ref{fig:basred} produces $t$. 
 The reduction sequence $\xrightarrow{\LRP,*}$  is the reflexive, transitive closure, 
  $\xrightarrow{\LRP,+}$ is the  transitive closure of $\xrightarrow{\LRP}$ and 
  $\xrightarrow{\LRP,k}$ denotes $k$ $\xrightarrow{\LRP}$-steps.
\end{definition}

The labeling algorithm proceeds top-down in an expression, marks the demanded subexpressions and finally detects the 
reduction position. It also marks the target position for a copy operation (see (cp-e) as an example), and the indirection chains used in 
(case)- and (seq)-reductions. 
 In Fig.\ \ref{fig:basred}  we omit the  types in all rules with the exception of (Tbeta) for simplicity. 
Note that normal-order reduction is type safe.


\begin{definition} \label{def:whnf}
A weak head normal form (WHNF) in LRP is a value $v$, or an expression $\tletr~\iEnv~\tin~ v$, where $v$ is a value, or an expression 
    $\tletr~x_1=c\ \vv{t},\{x_i=x_{i-1}\}_{i=2}^m,\iEnv~\tin~ x_m$.
 An expression $s$ {\em converges} to an expression $t$ ($s\maycon t$ or $s\maycon$ if we do not need $t$) if $s\xrightarrow{\LRP,*} t$ where $t$ is a WHNF.
  Expression $s$ {\em diverges} ($s\maydiv$) if it does not converge. 
\end{definition}

\ignore{
\begin{remark}
Our definition of WHNF is slightly different from the WHNF definition in \cite{schmidt-schauss-schuetz-sabel:08} and other papers.
Here we view $\tletr~~x_1=\lambda y.s,\{x_i=x_{i-1}\}_{i=2}^m,\iEnv~\tin~ x_m$ as a WHNF, whereas other papers  make one reduction step
to obtain  $\tletr~~x_1=\lambda y.s,\{x_i=x_{i-1}\}_{i=2}^m,\iEnv~\tin~ \lambda y.s$, which is the resulting WHNF in those papers. 

There is no difference in convergence, divergence, contextual preorder, equivalence. Also the reduction length $\rln(.)$ that counts 
the (lbeta), (case), (seq) reductions is unchanged. Only counting the number of (cp)-reductions may be different for the same expression. 
\end{remark}
}

\begin{definition} \label{def:equi}
For $\LRP$-expressions $s,t$ of the same type $\tau$, $ s\leq_c t$ holds iff $\forall C[\cdot_\tau]:C[s]\maycon \Rightarrow C[t]\maycon$,
and $s\sim_c t$ holds iff $s\leq_ct$ and $t\leq_c s$. The relation $\leq_c$ is called {\em contextual preorder} and $\sim_c$ is called {\em contextual equivalence}.
\end{definition}

The following notions of reduction length are used for measuring the time behavior in $\LRP{}$.
\begin{definition} \label{def:rln}
For a closed $\LRP$-expression $s$ with $s\maycon s_0$, let $\rln(s)$ be the sum of all (lbeta)-, (case)- and (seq)-reduction steps in $s\maycon s_0$,
let $\rln_{\LCSC}(s)$ be the  sum of all a-reduction steps in $s\maycon s_0$ with $a \in \LCSC$, where $\LCSC = \{(lbeta), (cp), (case), (seq)\}$, 
and let   $\rlnall(s)$ be the total number of reduction steps, but not (Tbeta), in $s\maycon s_0$. 
\end{definition}

\begin{figure*}[t]
 \fbox{
\begin{minipage}{.98\textwidth}
$\begin{array}{@{}lll@{}}
 \mbox{(gc1)}~ &   \tletrxnk{\{x_i=s_i\}_{i=1}^n, \iEnv }{t} \to \tletrxnk{\iEnv}{t}& \mbox{ if  } \forall i: x_i\not\in\FV(t,\iEnv), 
       n > 0
\\[.2ex]
\mbox{(gc2)}      &  \tletrxnk{x_1 = s_1,\ldots,x_n = s_n}{t} \to t&\mbox{ if for all } i: x_i \not\in \FV(t) 
\end{array}$  
    \caption{Garbage collection transformation rules for $\LRPgc$}\label{fig:extra-transformation-rules-LRPgc}  
 \end{minipage}
}
\end{figure*}
\begin{figure*}[t]
 \fbox{
\begin{minipage}{.98\textwidth}
     $\begin{array}{l@{~}c@{~}l}
  \size(x) &= & 0\\
  \size(s~t) &= & 1+ \size(s) + \size(t)\\
  \size(\lambda x.s) & = & 1+ \size(s)\\
  \size(\tcase~e~\tof~\talt_1 \ldots \talt_n) & = & 1+ \size(e)   ~~+ \sum_{i=1}^n \size(\talt_i)\\
  \size((c~x_1 \ldots x_n)~\casepf~ e) & = & 1 + \size(e) \\
  \size(c~s_1 \ldots  s_n) & = & 1 + \sum \size(s_i)\\
  \size(\tseq~s_1~s_2) & = & 1 + \size(s_1) +  \size(s_2)\\
  \size(\tletr~x_1 = s_1, \ldots,  ~~~x_n = s_n~\tin~s) & = & \size(s)  + \sum \size(s_i)
  \end{array}$
  \caption{Definition of $\size{}$}\label{fig:space-size}
 \end{minipage}
}
\end{figure*}
%

\ignoreForAbstract{
\subsection{The Untyped Calculus $\LR$}

To be self contained, we give the necessary definitions and  connections between $\LRP$ and $\LR$ as these appear in 
  \cite{sabel-schmidt-schauss:2015:frank-55}.
The good news is that if (Tbeta)-reduction steps (that only manipulate types) are ignored, then this constitutes exactly the normal order reduction 
of the untyped expression.

%
 
 \begin{definition}
 The calculus $\LR$  is defined on the set of expressions that is generated by a grammar that is derived from the one in  Fig.\ \ref{fig:lrp} 
 by omitting the types in the expression, but keeping the type constructor $K$ at the $\tcase_K$ constructs.\\
 The type erasure function $\varepsilon: \LRP \to \LR$ 
maps $\LRP$-expressions to
$\LR$-expressions by removing the types, the type information and the $\Lambda$-construct. In particular:
$\varepsilon(s~\tau) = \varepsilon(s)$, $\varepsilon(\Lambda a.s) = \varepsilon(s)$, $\varepsilon(x::\rho) = x$, and  $\varepsilon(c::\rho) = c$. 
We also define the type erasure for reduction sequences.
\end{definition}
 
 Clearly, $\xrightarrow{\LRP}$-reduction steps are mapped by $\varepsilon$ to 
 $\LR$-normal-order reduction steps
 where exactly the $(Tbeta)$-reduction steps are omitted.
The translation $\varepsilon$ is adequate, but not fully abstract: 
 \begin{proposition}
The translation $\varepsilon$ is adequate:\\
$\varepsilon(e_1) \sim_c \varepsilon(e_2) \implies e_1 \sim_{c} e_2$.
  \end{proposition}
  
  It is not fully abstract (i.e. $e_1 \sim_{c} e_2$ does not imply $\varepsilon(e_1) \sim_c \varepsilon(e_2)$);
   an example will be the (caseId) transformation (see Section \ref{sec-caseid}).
  

\begin{definition} Let $s,t$ be two $\LRP$-expressions of the same type $\rho$.
The improvement relation $\preceq$ for $\LRP$ is defined as: 
Let $s \preceq t$  iff $s \sim_c t$ and for all contexts $C[\cdot::\rho]$:
if $C[s], C[t]$  are closed, then
$\rln(C[s]) \leq \rln(C[t])$. 
If  $s \preceq t$ and  $t \preceq s$, we write $s \approx t$.
\end{definition}
 
The following facts are valid and can easily be verified or found in the literature 
   \cite{sabel-schmidt-schauss:2015:frank-55,schmidtschauss-sabel-PPDP-extended:2016,schmidt-schauss-schuetz-sabel:08}:
   
\begin{theorem}\label{thm:length-facts} ~    
\begin{enumerate}
\item For a closed $\LRP$-expression $s$, the equations $\rln(s) = \rln(\varepsilon(s))$ and 
$\rln_{LCSC}(s) = \rln_{LCSC}(\varepsilon(s))$  hold.
\item The reduction rules (Fig.\ \ref{fig:basred}) and extra transformations (Figs. \ref{fig:extra-transformation-rules}, 
\ref{fig:special-sp-impr-transformation-rules}
\ref{fig:special-transformation-rules}) 
  in their typed forms 
      can also be used in $\LRP$. They are correct program transformations and (time-) improvements.
 \item If $s \xrightarrow{a} t$ where $a$ is a reduction rule in any context, then $\rln_{LCSC}(s) \ge \rln_{LCSC}(t)$
 \item If $s \xrightarrow{a} t$ where $a$ is an extra transformation in any context, then $\rln_{LCSC}(s) = \rln_{LCSC}(t)$.
\item Common subexpression elimination applied to well-typed expressions is a (time-) improvement in $\LRP$ 
   (\cite{schmidtschauss-sabel-PPDP:2015}). 
\end{enumerate} 
\end{theorem}

 }

\subsection{LRPgc: LRP with Garbage Collection}\label{sec:calculi-gc}

As extra reduction rule in the normal order reduction we add garbage collection   (gc), 
which is the union of (gc1) and (gc2), but restricted to the top letrec (see Fig.\ \ref{fig:extra-transformation-rules-LRPgc}). 
%

\begin{definition}\label{def:LRP-variant}   
We define  $\LRPgc$, which employs all the rules of $\LRP$  and (gc) (see Fig.\ \ref{fig:extra-transformation-rules-LRPgc}) as follows: 
Let $s$ be an $\LRP$-expression.
A {\em normal-order-gc (LRPgc) reduction step} $s \xrightarrow{\nogc} t$ is defined by two cases:
\begin{enumerate}
 \item\label{LRPgc-1} If a (gc)-transformation is applicable to $s$ in  the empty context, i.e. $s \xrightarrow{gc} t$,
 then $s \xrightarrow{\nogc} t$, where the maximal possible set of bindings in the top letrec-environment of $s$ is removed.
 \item \label{LRPgc-2} If  \eqref{LRPgc-1}  is not applicable and $s \xrightarrow{\LRP} t$, then $s \xrightarrow{\nogc} t$.
\end{enumerate}
 
A sequence of $\nogc$-reduction steps is called a
{\em normal-order-gc reduction  sequence or LRPgc-reduction sequence}.
A WHNF without $\xrightarrow{\nogc,gc}$-reduction possibility is called an {\em LRPgc-WHNF}.
If the LRPgc-reduction sequence of an expression $s$ halts with a LRPgc-WHNF, then we say $s$ {\em converges w.r.t.\  LRPgc},
denoted as $s\downarrow_{\nogc}$, or $s\downarrow$, if the calculus is clear from the context.
\end{definition}

The extension of $\LRP$-normal-order reduction by garbage collection steps does not change the convergence and correctness: 

\begin{theorem}\label{thm:LRP-LRPgc-equivalence}
The calculus $\LRP$   is  convergence-equivalent to $\LRPgc$. I.e. 
for all expressions $s$: $s{\downarrow} \iff s{\downarrow}_{\nogc}$.
Contextual equivalence and preorder are the same for $\LRP$ and $\LRPgc$.
\end{theorem}
 

%

\section{Definitions of Space Improvements}\label{sec:space-improvements} 

From now on we use the calculus $\LRPgc$ as defined in Definition \ref{def:LRP-variant}.  
 We define an adapted (weaker) size measure than the size of the syntax tree, 
 which is useful for measuring the maximal space required to reduce an expression to a WHNF.
 The size-measure omits certain components. This turns into an advantage, since it enables proofs for the  exact behavior 
 w.r.t.\  our space measure for a lot of transformations.

\begin{definition}
The size $\size(s)$ of an expression $s$ is defined in Fig.\ \ref{fig:space-size}.
\end{definition}

%
The $\size$-measure does not count variables, it counts letrec-bindings  only by the size of the bound expressions, 
and it ignores the type expressions and type annotations. A justification for this omission is that this corresponds to the 
size (number of nodes)  of the sharing graph of the whole program. 
A technical justification for defining $\size(x)$ as $0$ is that let-reduction rules do not change the size, and 
that this is compatible with the size in the machine language. For example, the bindings $x = y$ do not contribute to the size.
This is justified, since the abstract machine (\cite{dallmeyer-schmidtschauss:16}) does not create $x=y$ bindings
(not even implicit ones) and 
instead makes an immediate substitution.

\ignore{
The sizes $\size$ and $\synsize$ differ only by a constant factor:

 \begin{proposition}\label{prop:spmax-constant-deviation}
 Let $s$ be an $\LRP$-expression. If $s$ does not permit a garbage collection of any binding, and there are no $x=y$-bindings, then  
 $\synsize(s) \leq (\mathit{maxarity}+1)*\size(s)$ and $\size(s) \leq \synsize(s)$, where $\mathit{maxarity}$ is the maximum of $2$ and 
 the maximal arity of constructor symbols in the language.
 \end{proposition}
 \begin{proof}
 It is sufficient to check every subexpression using an inductive argument. 
 \end{proof}
 }

\begin{definition}\label{def:of-spmax}
The space  measure $\spmax(s)$  of the reduction of a closed expression $s$ is 
the maximum of those $\size(s_i)$, where  $s_i   \xrightarrow{\nogc} s_{i+1}$ is not a (gc), and where the reduction sequence is
$s = s_0 \xrightarrow{\nogc} s_1   \xrightarrow{\nogc} \cdots  \xrightarrow{\nogc} s_n$, and $s_n$ is a WHNF.
If $s \mustdiv$, then $\spmax(s)$ is defined as $\infty$.

 For a (partial) reduction sequence $\Red = s_1 \to  \cdots \to s_n$, we define $\spmax(\Red)$ $=$ $\max_i\{size(s_i) \mid$  $s_i \to s_{i+1}$ 
  is not a (gc), and also  $s_n$ is not $\LRPgc$-reduccible with a (gc)-reduction$\}$  
\end{definition}
Counting space only if there is no (LRPgc,gc) possible is consistent with the definition in  \cite{gustavsson-sands:01}.
It also has the effect of avoiding certain small and short peaks in the space usage. The advantage is a better correspondence with the abstract machine
and it leads to comprehensive results. 

\begin{definition}\label{def:space-improvement}
Let $s,t$ be two expressions with $s \sim_c t$ and $s{\maycon}$. Then $s$ is a {\em space-improvement} of $t$, 
$s \leq_{\spmax}~t$, iff
for all contexts $C$ such that $C[s]$, $C[t]$ are closed, $\spmax(C[s]) \leq  \spmax(C[t])$ holds. 
The expression $s$ is  {\em space-equivalent} to $t$, $s \sim_{\spmax}~t$, iff for all contexts $C$ such that $C[s]$, $C[t]$ are closed, 
  $\spmax(C[s]) = \spmax(C[t])$ holds. 
A transformation $\xrightarrow{\mathit{Trans}}$ is called a space-improvement (space-equivalence) if
  $s \xrightarrow{\mathit{Trans}} t$ implies that $t$ is a space-improvement of (space-equivalent to, respectively)\ $s$. 
\end{definition}

Note that $\leq_{\spmax}$ is a precongruence, i.e. it is transitive and  $s \leq_{\spmax}~t$ implies $C[s] \leq_{\spmax}~C[t]$, and that
$\sim_{\spmax}$ is a congruence. Note also that  $s \leq_{\spmax}~t$ implies  $\size(s) \leq \size(t)$, using $C = \lambda x.[.]$.

Let $s,t$ be two expressions with $s \sim_c t$ and $s{\maycon}$.  
The relation  $s \leq_{R,\spmax} t$ holds, provided the following holds. 
For all reduction contexts $R$ such that $R[s]$, $R[t]$ are closed, we have $\spmax(R[s]) \leq  \spmax(R[t])$. 
The relation  $s \sim_{R,\spmax} t$ holds iff  $s \leq_{R,\spmax} t$ and  $t \leq_{R,\spmax} s$.

\begin{lemma}[Context Lemma for Maximal Space Improvement]\label{lemma:context-sp-impr}  
In $\LRPgc$ the following holds:  If $\size(s) \leq \size(t)$, $\FV(s) \subseteq \FV(t)$, and  
$s \leq_{R,\spmax} t$,  then $s \leq_{spmax} t$.
\end{lemma}
\begin{proof}(Sketch \cite{schmidt-schauss-dallmeyer:2017:frank-57})~
 The proof is by generalizing the claim to multiple pairs $(s_i,t_i)$ of expressions in multicontexts $M$, i.e. by comparing
$M[s_1,\ldots,s_n]$  and $M[t_1,\ldots,t_n]$, where the assumptions must hold for all pairs $s_i,t_i$. 
The induction proof is (i) on the number of $\nogc$-reduction steps of $M[t_1,\ldots,t_n]$, 
and (ii)  on the number of holes of $M$. 
The various cases of reductions of $M[t_1,\ldots,t_n]$ are analyzed, and in all cases the claim can be  shown  using the induction hypothesis.

Note that the proof technique would not work for call-by-name variants of the calculus. The reason is that substitution is incompatible with the proof technique.
\end{proof}

\begin{corollary}[Context Lemma for Maximal Space Equivalence]\label{lemma:context-sp-equiv} 
If 
$\size(s) = \size(t)$, $\FV(s) = \FV(t)$, and  
$s \sim_{R,\spmax} t$,  then $s \sim_{spmax} t$.
\end{corollary}

%
%
 The context lemmas    
 also hold  if the (stronger) condition $s \leq_{X,\spmax} t$, or $s \sim_{X,\spmax} t$, respectively,
  holds where $X$ means surface- or top-contexts. 
 
 
 We also consider useful program-transformations that are runtime optimizations, but may increase the space usage during runtime, 
and distinguish acceptable and bad behavior w.r.t.\  space usage. 
Transformations that applied in reduction contexts lead to a space increase of at most a fixed (additive) constant 
 are considered as controllable and safe, whereas the case that after the transformation the space increase may exceed 
any constant (depending on the usage of the expressions), is considered uncontrollable, and we say it is a space leak.

\begin{definition}\label{def:space-leak}
Let $\T$ be a transformation and let $s \xrightarrow{T}t$ be an instance with expressions $s,t$. 
\begin{enumerate}
\item\label{def:space-leak-1}  
We say that the $s \xrightarrow{T}t$ is {\em space-safe up to  the constant $c$},  
if for all reduction contexts $R$:
$\spmax(R[t]) \leq  c+\spmax(R[s])$. 
\item  If for some $c$, \eqref{def:space-leak-1} holds for all instances  $s \xrightarrow{T}t$, then we say $T$ is {\em space-safe up to  the constant $c$}.  
\item The transformation $s \xrightarrow{T}t$ is a {\em space leak}, iff for every $b \in \bbbR$, there is a reduction context $R$, such that
$\spmax(R[t]) \geq   b+\spmax(R[s])$.   
\item If there is one instance $s \xrightarrow{T}t$ that is a space leak, then we also say $T$ is a {\em space leak}. 
\end{enumerate} 
\end{definition}
\noindent This (simplistic) definition is a first criterion for a classification of transformations.   
 Definition  \ref{def:space-leak} for a classification of 
 transformations makes sense insofar as space-improvements are not space leaks and 
 space leaks cannot be space improvements.

 We will see below that there are examples of transformations that are not space-improvements but are space-safe up to an additive  constant, 
 and there are also transformations
 that are improvements w.r.t.\  runtime, but space leaks, like (cp),  (cse), and (soec). 

\begin{figure*}[t]
 \fbox{
\begin{minipage}{.98\textwidth}  
$\begin{array}{@{}ll@{}}
 \mbox{(cpx-in)}~ &   \tletrx{x = y,   \iEnv }{C[x]} 
       \to \tletrx{x = y,  \iEnv }{C[y]} \quad \mbox{ where $y$ is a variable and } x \not= y\\
 \mbox{(cpx-e)}~ &   \tletrx{x = y, z = C[x], \iEnv }{t}  
       \to \tletrx{x = y,  z = C[y], \iEnv}{t}   \hspace*{1cm} \mbox{(same as above)}\\ 

 \mbox{(cpcx-in)} &   \tletrx{x = c~\vv{t}, \iEnv }{C[x]}  
               \to  \tletrx{x = c~\vv{y}, \bchainN{y}{t}{\ari(c)}, \iEnv}{C[c~ \vv{y}]} \\  
 \mbox{(cpcx-e)} &   \tletrx{x = c ~\vv{t}, z = C[x], \iEnv }{t} 
        \to \tletrx{x = c ~\vv{y}, \bchainN{y}{t}{\ari(c)},  z = C[c~ \vv{y}],\iEnv}{t} \\

 \mbox{(abs)} & \tletrx{x = c~\vv{t},  \iEnv~ }{s}  
   \to \tletrx{x = c~\vv{x}, \bchainN{x}{t}{\ari(c)},  \iEnv~}{s} 
    \hspace*{5mm}  \mbox{where $\ari(c) \ge 1$} \\

  \mbox{(abse)} &  (c~\vv{t}) 
    \to \tletrx{\bchainN{x}{t}{\ari(c)}}{c~\vv{x}} 
       \hspace*{1cm} \mbox{where $\ari(c) \ge 1$} \\
 
\mbox{(xch)} & \tletrx{x = t, y = x, \iEnv}{r}~\to~ \tletrx{y = t, x = y, \iEnv}{r}  \\

\mbox{(ucp1)}& \tletrx{\iEnv, x = t}{S[x]} \to  \tletrx{\iEnv}{S[t]} \\
\mbox{(ucp2)}  &\tletrx{\iEnv, x = t, y = S[x]}{r}  \to    \tletrx{\iEnv, y = S[t]}{r} \\
\mbox{(ucp3)}      & \tletrx{x = t}{S[x]}  \to   S[t] \hspace*{1cm}  \mbox{where in the three (ucp)-rules, $x$ has at most} \\
    &  \qquad \mbox{ one occurrence in $S[x]$, no occurrence in $\iEnv,t,r$; 
      and $S$ is a surface context.} 
 \end{array}$ 
 \caption{Extra transformation rules}\label{fig:extra-transformation-rules}  
  \end{minipage}
}
\end{figure*}
%
 \begin{figure*}[t]
 \fbox{
\begin{minipage}{.98\textwidth}  
 $\begin{array}{@{}ll@{}}
  \mbox{(case-cx)}  & \tletrx{x = (c_{T,j}~x_1 \ldots x_n),\iEnv} {C[\tcase_T~x~((c_{T,j}~y_1 \ldots y_n) \casepf s)~alts]} \\
      &  \qquad   \to ~\tletrec ~x = (c_{T,j}~x_1 \ldots x_n),\iEnv  
           \quad \tin~C[\tletrx{y_1 = x_1,\ldots, y_n = x_n}{s}] \\
\mbox{(case-cx)}  & \tletrec ~ x = (c_{T,j}~x_1 \ldots x_n),\iEnv,  
     \quad  y =C[\tcase_T~x~((c_{T,j}~y_1 \ldots y_n) \casepf s)~alts]~\tin~r  \\
      &   \qquad    \to ~\tletrec~x = (c~x_1 \ldots x_n),\iEnv,   
         \quad  y = C[\tletrx{y_1 = x_1,\ldots, y_n = x_n}{s}]~\tin~r\\
 \mbox{(case-cx)}  & \mbox{in all other cases:  like (case)} \\
\mbox{(case*)}  & \mbox{is defined as (case) if the scrutinized data expression is of the form  $(c~s_1 \ldots s_n)$, } \\
      & \qquad  \mbox{where $(s_1,\ldots,s_n)$ is not a tuple of different variables, and otherwise it is \mbox{(case-cx)}  } \\
        \mbox{(\gceq)} & \tletrec~x = y, y = s,\iEnv~\tin~r~ \to~\tletrec~y = s,\iEnv~\tin~r \qquad \mbox{ where $x \not\in \FV(s,\iEnv,r)$,}\\
       & \hspace*{1cm} \mbox{ and $y = s$ cannot be garbage collected }\\
  \mbox{(caseId)}& (\tcase_K~s~(pat_1 \casepf pat_1) \ldots (pat_{|D_K|} \casepf pat_{|D_K|})) \to s  \quad \\
  \end{array}$
 \caption{Variations of transformation rules (space improvements)}\label{fig:special-space-improving-transformation-rules}
  \end{minipage}
}
\end{figure*}
\begin{figure*}[t]
 \fbox{
\begin{minipage}{.98\textwidth}  
 $\begin{array}{@{}ll@{}}
  \mbox{(cpS)}~ &   \mbox{is (cp) restricted such that only surface contexts $S$ for the target context $C$ are permitted} \\
     \mbox{(cpcxT)}~ &   \mbox{is (cpcx) restricted such that only top contexts $T$ for the target context $C$ are permitted} \\
   \mbox{(cse)} & \tletrec~x = s, y = s,\iEnv ~\tin~r   ~ \to~ \tletrec~x = s, \iEnv[x/y] ~\tin~r[x/y] \mbox{ where $x \not\in \FV(s)$}  \\ 
   \mbox{(soec)} & \mbox{changing the sequence of evaluation due to strictness knowledge by inserting \tseq.}\\ 
\end{array}$
 \caption{Some special transformation rules (space-worsening)}\label{fig:special-space-worsening-transformation-rules} 

  \end{minipage}
}
\end{figure*}
 
 \section{Space-Safe and Unsafe Transformations}\label{sec:space-safe-transformations}   

%

 More transformations are defined in Fig.\ \ref{fig:extra-transformation-rules}: (cpx) is the union of (cpx-in) and (cpx-e) and copies variables, 
 (cpcx) is the union of (cpcx-in) and (cpcx-e) 
 and copies constructor applications with variable-only-arguments,
 (abs), (abse) abstracts subexpressions by putting them in a binding environment, and (ucp) is the union of (ucp1), (ucp2), and (ucp3) and 
 is a (cp) into a unique occurrence of $x$, followed by a garbage collection. 
 Further transformations are defined and mentioned in Fig.\ \ref{fig:special-space-improving-transformation-rules}  and
\ref{fig:special-space-worsening-transformation-rules}: 
(case-cx) and (case*) are  variants of (case) which behave different if the tested expressions is of the form $(c~x_1 \ldots x_n)$ 
by optimizing the heap-bindings; (cpS) is (cp) where the target for copying is an $S$-context\footnote{$S$,$T$-contexts are defined in Section \ref{subsec:lrp-def}};  (cpcxT) is a variant of (cpcx), 
 where the target context is a $T$-context; (caseId) is a typed transformation that detects case-expressions that are trivial;
(cse) means common subexpression elimination;   (\gceq) is a  specialization of (gc) where a single binding $x=y$ in $s$ is removed,  
  where  $y$ is not free, and there is a binding for $y$ that cannot be garbage collected after the removal of $x = y$; and 
the transformation (soec) means a correct change only of the evaluation order by inserting \tseq-expressions, due to 
  strictness knowledge.
The notation like  $\xrightarrow{(T,(cpcxT))}$ means (cpcxT) applied in a $T$-context, and similar for others.
  


\subsection{On the Space-Safety of Transformations}
An overview of the results for max-space-improvements, -equivalences and space-worsening transformations are in the following theorem
where  further transformations are in Figs.\ \ref{fig:extra-transformation-rules-LRPgc}, \ref{fig:extra-transformation-rules}, 
\ref{fig:special-space-improving-transformation-rules} and \ref{fig:special-space-worsening-transformation-rules}.  
The proof technique for most of the proofs consists of computing complete sets of forking diagrams between transformation steps and the normal-order reduction 
steps and an appropriate induction proof on the length of reduction sequences  
(see \cite{schmidt-schauss-schuetz-sabel:08}  for more explanations),  where  computation of diagrams is simplified thanks to the context lemma
 (see \cite{schmidt-schauss-dallmeyer:2017:frank-57} for details). 

\quad 
 \begin{theorem}\label{main-theorem}  The following table shows the space-improvement and -safety properties of the mentioned transformations.\\
\begin{tabular}{|l|l|@{}}\hline
Improvement  & rules    \\  \hline
 $\succeq_{spmax}$  & 
   (lbeta), 
 (case),
 (seq),
 (lll),
 (gc),
 (case*), 
 (caseId) \\  
 $\sim_{spmax}$ &  
  (cpx),
 (abs), (abse),
 (xch),
 (ucp),
 (case-cx),
 (cpxT),
 (\gceq)\\ 
 $\not\succeq_{spmax}$   &  (cpcx), (cpS)\\ 
 space-safe up to 1  &  (T,(cpcxT))\\
 space-safe up to  $\size(v)$ &  (S,(cpS))  \\
  & \hspace{3mm} {\small where $v$ is the copied abstraction}\\
 space-leak &  (cp),  (cse), (soec)\\ 
 \hline
\end{tabular}
\end{theorem}
\begin{proof} Complete Proofs  for the space-safety can be found in  \cite{schmidt-schauss-dallmeyer:2017:frank-57},  
and sketches and remarks in the remainder of this section. 
%
%
As an example, we treat (cpx) in more detail:\\
%
%
%
{\bf Claim.} The transformation $(cpx)$ is a space-equivalence. \\
Due to the context lemma it is sufficient to check forking diagrams in top contexts, however, we permit that (cpx) may copy into arbitrary contexts.

An analysis of forking overlaps between LRPgc-reductions and $(cpx)$-transformations in top contexts shows that the following set of three diagrams is complete, where 
all concrete  (cpx)-transformations in a diagram copy from the same binding $x = y$: \\[-2mm]
%
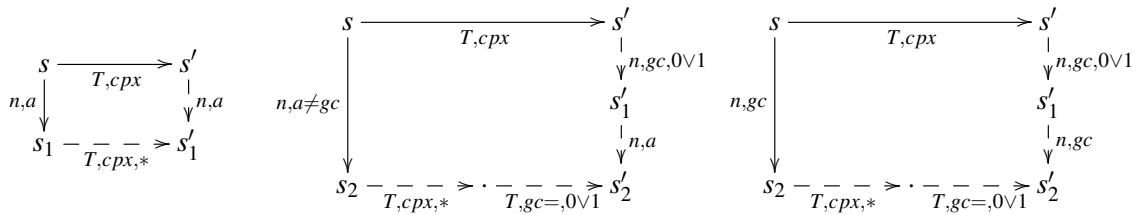
\begin{tabular}{llll}
\begin{minipage}{0.2\textwidth}
\[\xymatrix@R=5mm@C=14mm{
   s \ar[r]_{\T,cpx} \ar[d]_{n,a} &  s'  \ar@{-->}[d]^{n,a}
     \\
   s_1\ar@{-->}[r]_{\T,cpx,*} & s_1'
}\\
\]
\end{minipage}
&
\begin{minipage}{0.35\textwidth}
\[\xymatrix@R=5mm@C=14mm{
   s \ar[rr]_{\T,cpx} \ar[dd]_{n,a\not=gc} &&  s'  \ar@{-->}[d]^{n,gc,0\vee 1}
     \\
      & &    s_1'  \ar@{-->}[d]^{n,a}
      \\
   s_2\ar@{-->}[r]_{\T,cpx,*} & \cdot \ar@{-->}[r]_{\T,\gceq,0\vee 1} & s_2'
}\\
\]
\end{minipage}
&
\begin{minipage}{0.35\textwidth}
\[\xymatrix@R=5mm@C=14mm{
   s \ar[rr]_{\T,cpx} \ar[dd]_{n,gc} &&  s'  \ar@{-->}[d]^{n,gc,0\vee 1}
     \\
      & &    s_1'  \ar@{-->}[d]^{n,gc}
      \\
   s_2\ar@{-->}[r]_{\T,cpx,*} & \cdot \ar@{-->}[r]_{\T,\gceq,0\vee 1} & s_2'
}\\
\]
\end{minipage}
\end{tabular}

We also need the diagram-property that  $s_1 \xleftarrow{n,a} s \xrightarrow{\T,\gceq} s'$ can be joined by  
$s_1 \xrightarrow{\T,\gceq,0\vee 1} s_1' \xleftarrow{n,a} s'$. 
We will apply the context lemma for space equivalence (Proposition \ref{lemma:context-sp-equiv}), which also holds for $\T$-contexts.\\
Let $s_0 \xrightarrow{cpx} t_0$, and let $s = T[s_0]$ and $s' = T[t_0]$. Then  $\size(s) = \size(s')$ as well as $\FV(s) = \FV(s')$. 
We have to show $\spmax(s) = \spmax(s')$, 
which can be shown by an induction on the number of LRPgc-reductions of $T[s_0]$.  
The claim to be proved by induction is sharpened:  in addition  the number of LRPgc-reductions of $T[s_0]$ is not greater than for   $T[t_0]$. \\
Since $(cpx)$ as well as $(\gceq)$ do not change the size, we have the same maximal space usage  for $s$ and $s'$.
An application of  the context lemma for top contexts and for space equivalence finishes the proof.
\end{proof}

Note that a majority of the reasoning and proofs is done  in the untyped calculi $\LR$ and $\LRgc$  (see \cite{schmidt-schauss-dallmeyer:2017:frank-57}). 

 We  investigate the space-properties of  (cp):  Used as transformation (cpS) in an $S$-context it increases max-space at most 
by $\size(v)$ where $v$ is the copied abstraction;  and in general the size-increase can be bounded by $(\rln(s)+2)*\size(v)$ 
where $s$ is the initial expression.
This enables very useful estimations of the effects of optimizing transformations w.r.t. their max-space-behavior for the transformations 
 mentioned in this paper, 
in particular for optimizations by partial evaluation.

 \begin{proposition}\label{prop:cpS-S-maxspace} The following estimations hold for (cp) and  (cpS), where 
 $s \xrightarrow{cp} t$, and   where $v$ is the  copied abstraction:
\begin{enumerate}
  \item  The transformation $s \xrightarrow{(S,cpS)} t$ increases max-space at most by $\size(v)$.  
  \item  The transformation $s \xrightarrow{cp} t$ increases max-space at most by $(\rln(s)+2)*\size(v)$, i.e.  
   $\spmax(t) \leq (\rln(s)+2)*\size(v) + \spmax(s)$.
\end{enumerate}
  \end{proposition}

 %
  A consequence is that the space usage of several transformations $\xrightarrow{(S,cpS)}$ that are  space-safe up to  the additive constant $c$ 
  can be estimated:
  \begin{corollary}\label{thm:cpS-S-n-maxspace}
 Let $t$ be an expression. 
 If $t$ is transformed into $t'$ by an arbitrary number of space improvements that do not increase the size of abstractions, 
   including at most $n$ transformations that increase max-space by at most $c_i$ for $i = 1,\ldots,n$, and also by $m$ transformations 
  $\xrightarrow{(S,cpS)}$, then $\spmax(t') \le \spmax(t) +  (\sum c_i) + m\cdot V$, where $V$ is the maximum of the size of abstractions 
  in $t$.
  \end{corollary}
  \begin{proof}  
  This follows from Proposition \ref{prop:cpS-S-maxspace}   
  and since $\xrightarrow{(S,cpS)}$ does not increase the size of abstractions.  
 \end{proof} 
 
 \begin{remark}  
 Using (cp) as transformation  with general contexts for the target, for example copying into an abstraction, 
 may induce a space-leak, but see Proposition \ref{prop:cpS-S-maxspace}.     
 More exactly,  the max-space of a reduction sequence may increase linearly with the number of reduction steps, 
 and exponentially with the number of applications of the (cp)-transformation.  Examples of this behavior can be constructed as in 
  Example \ref{example:maxspace-unbounded}.
 
 Note that there are instances of (cp) that behave much better, for example versions of inlining (see below in Section \ref{sec:inlining}),
 or  if the copied abstraction can be garbage collected after (cp) or transformed further,  
 and also  the special case of (ucp)-transformations. 
 
 \end{remark}

 \subsection{Specific Examples and Comparison with Previous Work}\label{sec:specific}
  Now we explain several examples and compare with related work.
 
 \begin{example}\label{example:maxspace-unbounded} 
 We show that common subexpression elimination is a space leak.  
 We reuse an example which is similar to the example in  \cite{bakewell-runciman-ppdp:2000}.
 The expression is given in a Haskell-like notation, using integers, but can also be defined in $\LRPgc$:
 $s :=  \tif~(\texttt{last}~ [1..n]) > 0\ \tthen\ [1..n]\ \telse\ \tnil$, where
  $[1..n]$ is the expression that lazily generates a list $[1,\ldots,n]$. 
 The evaluation of $s$ expands the list  until the last element is generated and then 
 evaluates the same expression to obtain $1:[2..n]$. Due to eager garbage collection, it is not hard to see that the evaluation sequence
 requires constant max-space, independent of $n$ (assuming constant space for integers). Note that this evaluation will also generate 
   indirection chains of the form $\ldots, x_1 = x_2, x_2 = x_3, \ldots$, which are ignored by our space measure.  
 As shown in  \cite{dallmeyer-schmidtschauss:16} an evaluation on an abstract machine will really use constant space,
 if shortening indirections is performed by the abstract machine. 

 Now let $s' = \tletr~x = [1..n]~\tin\ \tif ~(\texttt{last}~x) > 0\  \tthen\ ~x\ \telse\  \tnil$. The evaluation of $s'$ behaves different to $s$: it first evaluates the list, and stores it 
 in full length, and then the second expression will be evaluated with an already evaluated list.
 The size required is a linear function in $n$. Seen from a complexity point of view, there is no real bound on this max-space increase: 
 the example can be adapted using any computable function $f$ on $n$ by modifying the list to $[1..f(n)]$.
 Obviously this example is a  space leak according to our definition, where the reduction contexts contains the list definition.
 \end{example}

%
There may be instances of common subexpression elimination which are not space leaks, however,  
    we leave the development of corresponding analyzes for future research.   
%

 The example and arguments in  \cite{bakewell-runciman-ppdp:2000} show that correctly changing the sequence of evaluations
 may be a transformation that is a space leak: this means that (soec) is classified as a space leak. 
 
 
%

%
%
%


\ignore{
The following may be useful if backward transformations enable other optimizations: 
\begin{proposition}
The transformation $t \to s$, where $s \xrightarrow{S,a} t$ where $a$ is a reduction rule from Fig.\,\ref{fig:basred} that is not a (cp)
 is space-safe up
to the constant $c = \spmax(s) - \spmax(t)$.   
\end{proposition}
}

\begin{example}
  An example that illustrates the definitions and may contribute to the discussion on  the boundaries between space-safe and -unsafe transformation is the following:
  Let $s =  \ttrue$ and $t = (id~\ttrue)$. Then clearly  $s \leq_{spmax} t$, and the transformation $s \to t$ is  space-safe.
   Let $s' = \lambda x. s$ and $t' = \lambda x. t$. Then $s' \to t'$  is a space leak according to our definition:\\
   Let $R = (\tletr~ y = [\cdot], z = r_n~\tin~({\tt and}~z)~ \&\&~ ({\tt last}~z))$, where $r_n$ is the list $[(y~0),\ldots,(y~0)]$ of length $n$, 
   {\tt and} is the 
   function that computes the logical conjunction of all list entries, and $\&\&$ is the logical conjunction. 
   Then  the difference  $\spmax(R[t])-\spmax(R[s])$
   is a linear function in $n$ that exceeds all bounds, hence $s' \to t'$ is a space leak.  
\end{example}

  {\bf Associativity of append.}\ \   In \cite{gustavsson-sands:01}, the re-bracketing of  $((xs~ \texttt{++}~ys)~\texttt{++}~zs)$ was analyzed, and the results
 had to use several variants of their improvement orderings; in particular their observation of stack  and heap space made the 
 analysis rather complex.
 We got results that are easier to obtain and to grasp due to our relaxed  measure of space.
  
 Our analysis of applying the associative law to the recursively defined append function $\texttt{++}$ shows that 
 $((xs~ \texttt{++}~ys)~\texttt{++}~zs)\ge_{spmax} (xs~ \texttt{++}~(ys~\texttt{++}~zs))$, where $xs, ys, zs$ are variables. 
 We know that  the two expressions are contextually equivalent.  The proof uses the context lemma for space improvement 
 and in particular the space-equivalence of (ucp) which allows to inline  uniquely used bindings, and 
   an induction argument. The exact analysis shows that within reduction contexts, 
   which exactly enforce the evaluation of the spine of the lists like $(\texttt{last}~[\cdot|)$, 
   the $\spmax$-difference is exactly $4$. However, for example in a reduction context $(\tseq~(\texttt{last}~[\cdot])~s)$, where 
   the evaluation of $s$ requires (much) more space than the expression $(\texttt{last}~[s])$, there is no max-space difference, 
   since we analyze the maximally used space. 
   The general estimation is that in reduction contexts $R$, we have 
    $\spmax(R[((xs~ \texttt{++}~ys)~\texttt{++}~zs)]) \leq 4+\spmax(R[(xs~ \texttt{++}~(ys~\texttt{++}~zs))])$.

 For the {\bf three sum-of-list-examples}  in \cite{gustavsson-sands:01}, the analysis using our size-measure results in comparable conclusions:
 they compare three functions: a plain recursively defined $\texttt{sum}$ of a list of numbers,   
 the tail-recursive function $\texttt{sum}'$ with a non-strictly used accumulator  and the tail-recursive  $\texttt{sum}''$ with a strictly used accumulator for the result.
 
  $\texttt{sum}$  requires space linear in the length of the list, and the same holds for   $\texttt{sum}'$.  
 However, \texttt{sum},  $\texttt{sum}'$  and $\texttt{sum}''$ as functions are not related by any improvement relation due to the change in
 the evaluation order of the spine and elements of the argument list, in case the list is not completely evaluated. 
 In the latter case transforming one into the other may indeed be a space-leak, independent of the length of the list since it would be an instance of (soec). 
 
 {\bf (weak-value-beta) in Fig.\ 2 in \cite{gustavsson-sands:99}}:  As a further comparison we check and compare our results (see  Proposition \ref{prop:cpS-S-maxspace})
  with  those for weak improvement in Fig.\ 2 in \cite{gustavsson-sands:99}: 
 the claim on (weak-value-beta) there appears to be  practically almost useless, (at least for a special case): copying once indeed can only increase the space by a 
 linear function in the size of the program (and as parameter the number of reductions in our formulation (see Prop.\ \ref{prop:cpS-S-maxspace})),
  even copying into an abstraction is permitted. However, repeating (weak-value-beta) $n$-times may increase the program exponentially (in $n$)
 by repeated doubling. The transformation rule in \cite{gustavsson-sands:99} permits $\tletr~x = V[x] ~\tin~C[x] \to \tletr~x = V[V[x]] ~\tin~C[V[x]]$ 
 $\to \tletr~x = V[V[V[V[x]]]] ~\tin~C[V[V[V[x]]]]$, where $V$ is a value as context. 
 Hence, a sequence of  several weak improvement steps is not space-safe in the intuitive sense.
 According to our definition it is a space leak for this particular example.  
 
 Our foundations allow to improve the claims on the space-properties of the two last let-shuffling rules of \cite{gustavsson-sands:99}, 
 which are (strong) space-improvements w.r.t.\  our measure and definitions,
 since we have proved  that (lll) is a space equivalence. 
 {\bf Typed Transformations}\ \ 
The rule (caseId) is also the heart of other type-dependent transformations,  
which are also only correct under typing, and is a space-improvement. Examples of more general transformations of a similar kind are:  
 $(\texttt{map}~\lambda x.x) \to \texttt{id}$, $\texttt{filter}~(\lambda x.\ttrue) \to \texttt{id}$, and  $\texttt{foldr}~(:)~[] \to \texttt{id}$,
where we refer to the 
usual Haskell-functions and constructors. Note that these transformations are not correct in the untyped calculi.
 
\begin{figure*}[t] \small
\fbox{
%
   \begin{minipage}{.98\textwidth}
  Without inlining:\\
  $\begin{array}{ll}
  \foldl{} &= \lambda f,z,xs.\tcase~xs~\tof~\{([]~\casepf~z)~((y:ys)~\casepf~\foldl~f~(f~z~y)~ys)\}\\
  \foldls{} &= \lambda f,z,xs.\tcase~xs~\tof~\{([]~\casepf~z)~((y:ys)~\casepf~\tlet~w = (f~z~y)~\tin~\tseq~w~(\foldls~f~w~ys))\}\\
  \foldr{} &= \lambda f,z,xs.\tcase~xs~\tof~\{([]~\casepf~z)~((y:ys)~\casepf~f\ y\ (\foldr~f~z~ys))\}
  \end{array}$
  \\With inlining using \texttt{xor} as function:\\
  $\begin{array}{ll}
  \texttt{xor} &= \lambda x,y.\tcase~x~\tof~\{(\texttt{True}~\casepf~\tcase~y~\tof~\{(\texttt{True}~\casepf~\texttt{False})~(\texttt{False}~\casepf~\texttt{True})\})~(\texttt{False}~\casepf~y)\}\\
  \foldl{} &= \lambda f,z,xs.\tcase~xs~\tof~\{([]~\casepf~z)~\\
     & \hspace*{34.5mm} ((y:ys)~\casepf~\foldl~f~(\tcase~z~\tof\,\{(\texttt{True}~\casepf~\tcase~y~\tof\,\{(\texttt{True}~\casepf~\texttt{False})\\
     & \hspace*{114.5mm}(\texttt{False}~\casepf~\texttt{True})\})\\
     & \hspace*{83mm}(\texttt{False}~\casepf~y)\})~ ys)\}

  \end{array}$
  \\[1mm]
  \hphantom{a}The inlining of $\foldls{}$ and $\foldr{}$ is analogous to the inlining of $\foldl{}$.

  \caption{Definitions of $\foldl{}$, $\foldls{}$ and $\foldr{}$}\label{fig:fold-def}
  \end{minipage}
}
\end{figure*}
 
 {\bf Translating into Machine Language}\ \ 
 An efficient implementation of the evaluation of programs or program expressions first translates expressions
 into a machine format that can be executed by an abstract machine.  
 We consider a translation into a variant of the Sestoft machine \cite{sestoft:97,moran-sands:99} extended by (eager) garbage collection. 
The translation $\psi$  into machine expressions (see also \cite{schmidtschauss-sabel-PPDP:2015})
in particular translates   $\psi(s\ t)$  to   $\tletr~ y=\psi(t)~\tin~ (\psi(s)\ y)$, which is the same as an inverse (ucp). 
Our results, in particular the results on (ucp) (Thm.\ \ref{main-theorem}), show that the complete translation $\psi$  is a space-equivalence.
Note that the abstract machine uses extra data structures for evaluation.

 \begin{figure*}[t] \small
\fbox{
\begin{minipage}{0.9\textwidth}
\arraycolsep=7pt
$
\begin{array}{l|rrrrrrrrrr}
k & 100 & 200 & 300 & 400 & 500 & 600 & 700 & 800 & 900 & 1000\\
\hline\hline
 & \multicolumn{10}{l}{\texttt{foldl False xor (take k lst)}}\\
\hline
\rln     & 1214 & 2414 & 3614 & 4814 & 6014 & 7214 & 8414 & 9614 & 10814 & 12014\\
\spmax{} & 825 & 1625 & 2425 & 3225 & 4025 & 4825 & 5625 & 6425 & 7225 & 8025\\
\hline\hline
 & \multicolumn{10}{l}{\text{after inlining:}}\\
\hline
\rln     & 1012 & 2012 & 3012 & 4012 & 5012 & 6012 & 7012 & 8012 & 9012 & 10012\\
\spmax{} & 882 & 1682 & 2482 & 3282 & 4082 & 4882 & 5682 & 6482 & 7282 & 8082\\
\hline\hline
 & \multicolumn{10}{l}{\foldls \texttt{ False xor (take k lst)}}\\
\hline
\rln     & 1315 & 2615 & 3915 & 5215 & 6515 & 7815 & 9115 & 10415 & 11715 & 13015\\
\spmax{} & 63 & 63 & 63 & 63 & 63 & 63 & 63 & 63 & 63 & 63\\
\hline\hline
 & \multicolumn{10}{l}{\text{after inlining:}}\\
\hline
\rln     & 1113 & 2213 & 3313 & 4413 & 5513 & 6613 & 7713 & 8813 & 9913 & 11013\\
\spmax{} & 75 & 75 & 75 & 75 & 75 & 75 & 75 & 75 & 75 & 75\\
\hline\hline
 & \multicolumn{10}{l}{\texttt{foldr False xor (take k lst)}}\\
\hline
\rln     & 1115 & 2215 & 3315 & 4415 & 5515 & 6615 & 7715 & 8815 & 9915 & 11015\\
\spmax{} & 66 & 66 & 66 & 66 & 66 & 66 & 66 & 66 & 66 & 66\\
\hline
 & \multicolumn{10}{l}{\text{after inlining:}}\\
\hline
\rln     & 913 & 1813 & 2713 & 3613 & 4513 & 5413 & 6313 & 7213 & 8113 & 9013\\
\spmax{} & 84 & 84 & 84 & 84 & 84 & 84 & 84 & 84 & 84 & 84\\
\end{array}
$
\end{minipage}
}
\caption{Experimental Results }\label{fig:fold-experiments}
\end{figure*}

 \section{Experimental Analysis of Inlining and \texttt{fold} Using the Tool LRPi}\label{sec:inlining}
We use our interpreter LRPi that executes $\LRPgc$-programs using an abstract machine approach and measures the runtime and space behavior 
(for more details see \cite{dallmeyer-schmidtschauss:16}) and apply it to several fold-variants. 
The correctness of the space measurement of the abstract machine is described in   
 \cite{sabel-schmidt-schauss:2015:frank-55} (see Thm.\, \ref{main-theorem}). 

We analyze the space behavior of \texttt{fold} using exclusive-or as function, \texttt{False} as neutral element and a list \texttt{lst} starting with a 
single \texttt{True} followed by $k-1$ \texttt{False}-elements generated using a take-function/list generator approach. 
We already compared the three fold variants concerning runtime and space consumption with each other in this scenario in \cite{dallmeyer-schmidtschauss:16}, 
but now we focus on the impact of inlining, i.e. how does inlining affect the space consumption in the same scenario?

The current version of LRPi uses Peano-encoding for positive integers, but treats arbitrary Peano numbers as of size $1$, 
which makes it is easier to analyze the results. Hence the current statistics
 differs from that in 
\cite{dallmeyer-schmidtschauss:16}.   
The \texttt{fold}-functions are defined in Fig.~\ref{fig:fold-def}. 
Inlining copies the defining lambda-expression for
 \texttt{xor} to a call site and then applies (lbeta), (ucp), (cpx), (gc) perhaps several times to obtain the inlined definitions in Fig.\ \ref{fig:fold-def}.
 In order to keep the experiment simple and interpretable, we omit further obvious optimizations. 
 Since inlining copies into an abstraction (in addition into a recursive definition), our theoretical results do not give good guarantees on the space behavior
 and also do not preclude that the transformation might be a space leak.
 Here further research is needed.

 Fig.\ \ref{fig:fold-experiments}  uses  Haskell-notation where $k$ is the length of the input list, \texttt{rln} is the 
runtime measure, i.e. the number of (essential) reduction steps (see Def.\ \ref{def:rln}), $\spmax$ is our space measure (see Def.\ \ref{def:of-spmax}). 
 
For \texttt{foldl} the runtime decreases linearly after inlining, since this decreases the number of reduction steps by a constant for 
each list element. In contrast the inlining increases the space consumption by a constant.
Space consumption is linear in the length of the input list, which  is caused by the left-associativity of $\foldl{}$ since we get 
a linear number of nested \texttt{case}-expressions caused by the \texttt{xor}. 
The constant increase of space consumption after inlining is caused by the constant additional space that is needed by the inlined \texttt{xor}-function.

If we use the strict variant of \texttt{foldl} (i.e. $\foldls$), then the accumulator is evaluated each time and therefore no nested \texttt{case}-expressions are constructed. 
As expected we see that the space consumption is constant and that the runtime again improves linearly.

For \texttt{foldr}  the inlining improves the runtime linearly. Moreover the space consumption also only increases by a constant (similar to the
  $\foldls$-variant).  
The reason  is the right-associativity of \texttt{foldr}: since \texttt{xor} is strict in the first argument, \texttt{foldr} runs over the whole list, 
but depending on the left argument, \texttt{xor} either evaluates the second argument or returns the argument. Since the list is lazily generated 
and contains only \texttt{False}-elements (up to one occurrence), each element gets directly generated and consumed and therefore only constant space is needed.

The example suggests that it is a good idea to invest (a bit of) space for time, since for \texttt{foldl'} and \texttt{foldr} the runtime improves linearly while the space consumption only 
increases by a constant.
Our experiment shows a nice behavior in the considered empty context, but does not show the behavior in other contexts, or other uses of the functions. 


\section{Conclusion and Future Research}

We successfully derived results on the space behavior of transformations in lazy functional languages, by defining a relaxed space measure and 
reasoning about space improvements and space equivalences. We developed and justified a criterion for classifying transformations  as space safe or space leaks.
An impact of our results could be a controlled runtime optimization during compile time by applying time-improving transformations, but taking into account the knowledge
of their impact on the max-space usage. We contributed by detailing and refining this knowledge for call-by-need functional languages.

Future work is to extend the analysis of transformations  to larger and more complex transformations in the polymorphic typed setting.
A generalisation of top contexts to surface contexts is also of value for 
several transformations. 
To develop methods and justifications for space-improvements involving recursive definitions is left for future work, as well as 
the exploration of  the use of space-ticks as in \cite{gustavsson-sands:01} to improve the computation of estimations.
\paragraph*{Acknowledgements}

We thank David Sabel for lots of discussions and helpful hints,    and reviewers for their constructive comments.

\bibliographystyle{eptcs}

\end{document}